\documentclass[a4paper,UKenglish,cleveref, autoref, thm-restate]{lipics-v2021}
\pdfoutput=1

\title{The Simplest Proof of Parikh's Theorem via Derivation Trees} 

\author{Alexander {Rubtsov}}{National Research University Higher School of Economics \and Moscow Institute of Physics and Technology  }{rubtsov99@gmail.com}{https://orcid.org/0000-0001-8850-9749}{} 

\authorrunning{A.~Rubtsov} 

\ccsdesc[500]{Theory of computation~Grammars and context-free languages}

\keywords{Formal Languages, Context-Free Languages, Parikh's Theorem} 

\funding{The paper was supported by RFBR grant 20-01-00645}

\acknowledgements{I want to thank Dmitry Chistikov for the helpfull discussion of the paper}

\usepackage{amsthm}

\let\geq\geqslant

\let\es\varnothing

\def\NN{\mathbb N}

\let\epsilon\varepsilon
\let\eps\varepsilon

\def\abstractname{}

\usepackage{mathtools}

	\def\yields{\xRightarrow{{\scriptstyle *\; }}}

	\def\ders{ \vdash^{\!\!\! {}^{ *}}}
	\def\der{ \vdash }

\usepackage{mathrsfs}

\def\Lin{\mathrm{Lin}}

\usepackage{pgfplots}
\pgfplotsset{compat=1.15}
\usepackage{mathrsfs}
\usetikzlibrary{automata, positioning, shapes, snakes, arrows, arrows.meta, decorations.markings, matrix, cd}

\usepackage{multicol}

\hideLIPIcs
\nolinenumbers

\begin{document}

\maketitle

\begin{abstract}
	Parikh's theorem is a fundamental result of the formal language's theory. There had been published many proofs and many papers claimed to provide a simplified proof, but most of them are long and still complicated. We provide the proof that is really short, simple and discloses the nature of this fundamental result. We follow the technique closed to the original Parikh's paper and our proof is similar to the proof by Ryoma Sin'ya 2019, but we provide more detailed exposition and pretend to more simplicity as well. We achieve the simplicity via nonconstructivenes that allows us avoiding many difficulties met by other proofs.
\end{abstract}

\section{Introduction}

Parikh's theorem~\cite{ParikhOriginal} is a fundamental theorem of the formal language's theory. There had been published many proofs (see \cite{ESPARZA2011614}, \cite{KochSurvey} for the survey of different proofs and detailed exposition on the topic). Many papers claimed to provide a simplified proof (\cite{goldstine1977simplified} looks to be the most known), but most of them are long and still complicated. Despite really short and simple proofs are already known (e.g.,~\cite{ShallitSC}), papers with another proofs continues to be published~\cite{KogaIJFCS21}. We provide the proof that is really short, simple and discloses the nature of this fundamental result. Our proof is based on derivations trees so as the original Parikh's proof~\cite{ParikhOriginal}. As Parikh, we decompose derivation trees into small ones; we consider similar kinds of trees: ordinary derivation trees and auxiliary ones (with only nonterminal in the crown that is the same as the root). We get rid of duplicates in auxiliary trees and pump them. While our technique is similar to the original proof in general, it is simpler since we do not have restrictions on grammar and other technical issues (like considering derivation trees that contains nonterminals from a fixed subset). The simplicity of our construction is based on nonconstructivenes that allows us avoiding many technical issues. Also trees (in the decomposition) in our construction have linear height (in the number of nonterminals) while in the Parikh's construction trees have quadratic height. In our proof we generalize the idea of derivation from words to trees and translate the idea of the pumping lemma to the trees as well. It makes our construction clear enough to explain students during a lecture in a basic course of formal languages and automata theory, and so we hope that the detailed version will help to spread it in the community.

Our technique is similar to~\cite{ProofAlaTakahashi}. Unlike concise exposition in~\cite{ProofAlaTakahashi}, we provide more detailed exposition and explain the intuition of the proof. And we do not use auxiliary results for our proof (thanks to nonconstructivenes).



\section{Definitions}
We denote by $\NN$ non-negative integers. Let $\Sigma_k = \{a_1, \ldots, a_k\}$ be an alphabet, $k \geq 1$. We denote by $\Psi: \Sigma^*_k \to \NN^k $ the \emph{Parikh mapping} that maps a word $w$ to its \emph{Parikh's image} the vector $(|w|_{a_1}, |w|_{a_2}, \ldots, |w|_{a_k})$, where $|w|_{a_j}$ is the number of letters $a_j$ in $w$. We denote the Parikh's image of a language $L \subseteq \Sigma_k^*$ in the natural way: $$ \Psi(L) = \{ \Psi(w) \mid w \in L \}. $$

A set $S \subseteq \NN^k$ is \emph{linear} if there exist vectors $v_0, v_1, \ldots, v_m \in \NN^k$ such that
$$ S = \{ v_0 + v_1 t_1 + v_2 t_2 + \ldots v_m t_m \mid t_1, \ldots t_m \in \NN \}. $$
A set is \emph{semi-linear} if it is a union of finitely many linear sets.

\begin{theorem}[Parikh]
	For each context-free language $L$ the set $\Psi(L)$ is semi-linear.
\end{theorem}

We use classical notation for context-free grammars~\cite{hopcroft2001introduction}. We denote by $N$ the set of nonterminals, and denote nonterminals by capital letters. Small letters from the end of the alphabet denote words and Greek letters denote words over the alphabet $\Sigma\cup N$, called \emph{sentential forms}.


Our proof uses derivation trees and exploits the idea similar to the classical proof of the pumping lemma. We use derivation trees not only for words, but also for sentential forms of a special kind. We call a tree that corresponds to a derivation of the form $A \yields uAv$ a \emph{block tree}. We call derivation trees for words \emph{ground trees}. We say that a tree $T$ (a ground or a block one) is \emph{minimal} if it does not have a block tree as a subtree, i.e. $T $ does not have the form $ (S \yields x A z \yields xuAvz \yields xu\beta vz)$. Hereinafter, we \emph{describe} a tree by any derivation corresponding to the tree.


\begin{figure}[h]
	
	\begin{multicols}{3}
	
		\begin{tikzpicture}[line cap=round,line join=round,>=triangle 45,x=0.6cm,y=0.6cm]
		\clip(-0.05,-0.7) rectangle (9,7);
		\fill[line width=2pt,color=black,fill=white] (4,6) -- (0,0) -- (8,0) -- cycle; 
		\fill[line width=2pt,color=black,fill=black,fill opacity=0.1] (4,3) -- (2,0) -- (6,0) -- cycle;
		\fill[line width=2pt,color=black,fill=black,fill opacity=0.4] (4,1.5) -- (3,0) -- (5,0)  -- cycle;
		\draw [line width=2pt,color=black] (4,6)-- (0,0);
		\draw [line width=2pt,color=black] (0,0)-- (8,0);
		\draw [line width=2pt,color=black] (8,0)-- (4,6);
		\draw [line width=2pt,color=black] (4,3)-- (2,0);
		\draw [line width=2pt,color=black] (2,0)-- (6,0);
		\draw [line width=2pt,color=black] (6,0)-- (4,3);
		\draw [line width=2pt,color=black] (4,1.5)-- (3,0);
		\draw [line width=2pt,color=black] (3,0)-- (5,0);
		\draw [line width=2pt,color=black] (5,0)-- (4,1.5);
		\draw  (4.,6.4) node {$S$};

		\draw (4,3.4) node {$A$}; 
		\draw (4,1.9) node {$A$}; 

		\draw (1.00,-0.35) node {$x$}; 
		\draw (2.57,-0.35) node {$u$}; 
		\draw (4.09,-0.35) node {$y$}; 
		\draw (5.58,-0.35) node {$v$}; 
		\draw (7.10,-0.35) node {$z$}; 

		\end{tikzpicture}
		
		\columnbreak
								
				\hspace*{0.5cm}
				\begin{tikzpicture}[line cap=round,line join=round,>=triangle 45,x=0.6cm,y=0.6cm]
				\clip(-0.05,-0.7) rectangle (9,7);
				\fill[line width=2pt,color=black,fill=black,fill opacity=0.5] (4,3) -- (3,1.5) -- (5,1.5)  -- cycle;
				\draw [line width=2pt,color=black] (4,6)-- (0,0);
				\draw [line width=2pt,color=black] (0,0) -- (2,0);
				\draw [line width=2pt,color=black] (6,0) -- (4,3);
				\draw [line width=2pt,color=black] (6,0) -- (8,0);
				\draw [line width=2pt,color=black] (8,0)-- (4,6);
				\draw [line width=2pt,color=black] (4,3)-- (2,0);	
				\draw [line width=2pt,color=black] (6,0)-- (4,3);
				\draw [line width=2pt,color=black] (3,1.5)-- (5,1.5);
		\draw  (4.,6.4) node {$S$};

		\draw (4,3.4) node {$A$}; 


				\draw (4,1.1) node {$y$}; 

				\draw (1.00,-0.35) node {$x$}; 
				\draw (7.10,-0.35) node {$z$}; 

				\end{tikzpicture}
				
				\columnbreak
				
				
					\begin{tikzpicture}[line cap=round,line join=round,>=triangle 45,x=0.6cm,y=0.6cm]
					\clip(-0.05,-0.7) rectangle (9,7);
					\fill[line width=2pt,color=black,fill=black,fill opacity=0.1] (4,3) -- (2,0) 
					-- (3,0) -- (4,1.5) -- (5,0) -- (6,0) -- cycle;
					\draw [line width=2pt,color=black] (4,3)-- (2,0);
					\draw [line width=2pt,color=black] (2,0)-- (3,0);
					\draw [line width=2pt,color=black] (5,0)-- (6,0);					
					\draw [line width=2pt,color=black] (6,0)-- (4,3);
					\draw [line width=2pt,color=black] (4,1.5)-- (3,0);
					\draw [line width=2pt,color=black] (5,0)-- (4,1.5);


		\draw (4,3.4) node {$A$}; 

		\draw (4,1.9) node {$A$}; 
		\draw (2.57,-0.35) node {$u$}; 
		\draw (5.58,-0.35) node {$v$}; 

\end{tikzpicture}

	\end{multicols}
	\caption{Decomposition of tree $T$ into the pair of trees $(T_1, T_2)$}\label{fig::tree:decomposition}
\end{figure}
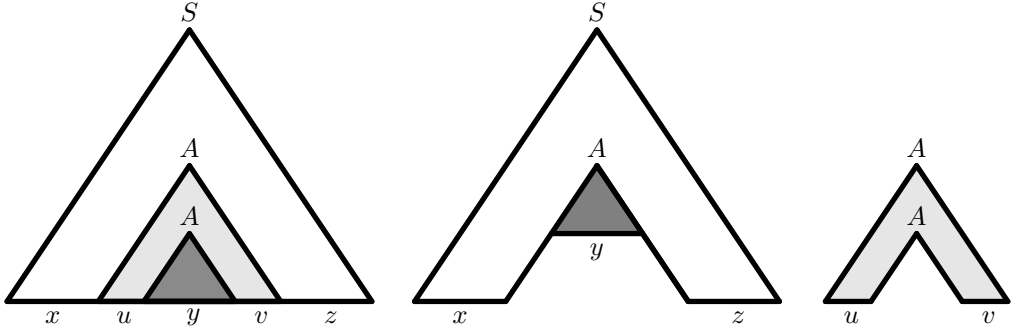

A non-minimal ground tree has the form $S \yields xAz \yields xuAvz  \yields xuyvz$. It can be \emph{decomposed} into a pair of the ground tree $S \yields xAz \yields xyz$ and the block tree $A \yields uAv$, see Fig.~\ref{fig::tree:decomposition}. Formally, we say that a tree $T$ is decomposed to the pair of trees $(T_1, T_2)$ if for some $B \in N : T = A \yields x B z \yields x u B v z  \yields x u \beta v z $, $T_1 = A \yields x B z \yields x \beta z$ and $T_2 = B \yields u B v $; here $\beta$ is either $A$ if $T$ is a block tree or $\beta \in \Sigma^*$ if $T$ is a ground tree (and in this case $A$ is the axiom). We say that $T$ \emph{can be composed} from $T_1$ and $T_2$. Note that $T_1$ and $T_2$ can be composed in several ways if $T_1$ has several $B$ nodes. We denote $$T_1 \circ T_2 = \{ T \mid T \text{ can be decomposed into } (T_1, T_2) \}$$

In addition to derivation of words we consider \emph{derivation of trees} as follows.

\begin{definition}\label{def:DerTree}    
 A derivation of a tree starts with a minimal ground tree and each derivation step leads to a ground tree as well. At a derivation step we choose a node $A$ and a minimal block tree $A \yields xAy$, replace the chosen node by the block tree; the subtree of the chosen node is glued into $A$ from the crown of the block tree. 
\end{definition}

So in the trees derivation a minimal ground tree has the role of the axiom (there can be several ones), and a replacement $A \to T_i$ where $T_i = (A \yields x A z)$ has the role of the production rule.

\begin{definition}\label{def:DerMultiset} 
	Let $S$ be a multiset of trees. We say that $S'$ is \emph{derived} from $S$ if $S'$ is obtained from $S$ by the replacement of two trees $T_1$ and $T_2$ by a tree $T \in T_1\circ T_2$. We denote it as $S \der S'$. We say that a multiset of trees $S$ is \emph{well-formed} such that $S \ders \{ T \}$ for some ground tree $T$.
\end{definition}

It is easy to see, that if $S$ contains only minimal trees and $S \ders \{ T \}$, then there exists a derivation of $T$ such that $S = \{T_0, T_1, \ldots, T_n\}$ where $T_0$ is the ground tree and $T_i, i \geq 1$ is the block tree that was chosen at the $i$-th derivation step.

We define the Parikh image for the trees in a natural way. $\Psi(T) = \Psi(w)$ if $T = (S \yields w)$ and $\Psi(T) = \Psi(xz)$ if $T = (A \yields xAz)$.
The following lemma directly follows from the definitions.

\begin{lemma}\label{Lemma:WellDefineParikhMap} \phantom{x}
\begin{itemize}
	\item If $T, T' \in T_1 \circ T_2$, then $\Psi(T) = \Psi(T') = \Psi(T_1) + \Psi(T_2)$.
	\item If $S \ders \{T\}$ and $S \ders \{T'\}$ then $\Psi(T) = \Psi(T')$.
\end{itemize}
	 
\end{lemma}

Lemma~\ref{Lemma:WellDefineParikhMap} implies that we can extend the definition of the Parikh map to well-formed multisets: $\Psi(S) = \sum_{T\in S}\Psi(T)$.

By the pigeonhole principle each minimal tree has depth at most $|N|-1$ (otherwise on the longest path from the root to a leaf there would be a repetition of nonterminals). Since (for a fixed grammar) each node of a tree has a bounded degree, there is only finitely many minimal trees. Let us enumerate them and denote the number of minimal trees by $m$. So, each multiset of minimal trees $S$ has a \emph{corresponding vector} $\vec{v} \in \NN^m$ where $\vec{v}_i$ is the number of occurrences of $T_i$ in $S$. 

\section{Proof}

We begin with the proof idea. Any word $w$ derived from a grammar $G$ has some derivation tree $T_w$ and $\Psi(T_w)= \Psi(w)$. We decompose $T_w$ into a multiset~$S_w$ of minimal trees ($S_w \ders T_w$). Denote by $T_g$ the ground tree from $S_w$. Obtain the set $S'_w$ from $S_w$ by removing repetitions. Note that $S'_w \subseteq M_G$ where $M_G$ is the finite set of minimal trees of grammar $G$. So we have defined a mapping $w \mapsto S'_w$ with a finite codomain $2^{M_G}$. For any multiset $S$ obtained from $S'_w$ by repetition of non-ground trees, there exists a derivation tree $T$ (of $G$) such that $S \ders T$ ($S$ is well-formed). Therefore $\Psi(T) \in \Psi(L(G))$. In other words, 
$$
 \left(\Psi(S'_w) + \sum\limits_{T' \in S'_w\setminus\{T_g\}} t_{T'}\cdot\Psi(T')\right)\in \Psi(L(G))
$$
no matter how we choose $t_{T'} \in \NN$ for each $T'$. Let $S$ be a well-formed multiset with a ground tree $T_g \in S$. Denote by 
\begin{equation}\label{eq:LinDef}
    \Lin(S) = \left\{ \Psi(S) + \sum\limits_{T'_i \in S\setminus\{T_g\}} t_i\cdot\Psi(T'_i) \quad\Bigg|\quad t_i \in \NN\right\}.    
\end{equation}    
The set $\Lin(S)$ is linear by the definition. Since there are only finitely many different sets~$S'_w$, the set~$\Psi(L(G))$ is a finite union of linear sets $\Lin(S'_w)$, so it is semilinear by the definition. The scheme of the proof is depicted in Fig.~\ref{fig::proofdiagram}

\tikzcdset{arrow style=tikz, diagrams={>=stealth}}

\begin{figure}[h]
    \begin{center}     
        \begin{tikzcd}
          L(G) \arrow[rrrrrr, "\Psi"]  \arrow[d] &&&&&& \Psi(L(G)) \\        
          \{T_w : w \in L(G)\} \arrow[d] \\
          \{S_w : w \in L(G)\} \arrow[d] \\
          \{S'_w : w \in L(G)\} \arrow[d] \\
          \{\Lin(S'_w) : w \in L(G)\} \arrow[uuuurrrrrr, "\cup"]
        \end{tikzcd}
      \end{center}
      \caption{Scheme of the proof}\label{fig::proofdiagram}
\end{figure}
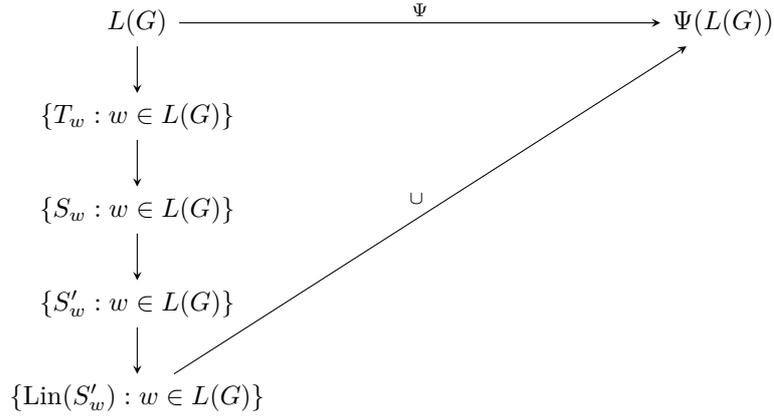

\subsection{Auxiliary Lemmas}

\begin{lemma}\label{lemma:BaseSet}
	Let $S$ be a well-formed multiset of minimal trees and $S''$ be a multiset. Denote by $\vec{v}$ and $\vec{v}''$ the corresponding vectors of $S$ and $S''$, respectively.  Suppose $\vec{v}_i > 0$  iff $\vec{v}_i'' > 0$; then  $S''$ is a well-formed multiset as well. 
\end{lemma}
\begin{proof}
    The statement is equivalent to the conjunction of two claims. The first one: if $\vec{v}_i > 0$ and $\vec{v}_i'' = \vec{v}_i + 1$ while $\vec{v}''_j = \vec{v}_j$ for $j\neq i$, then $S''$ is well-formed. The second one: if $\vec{v}_i > 1$ and $\vec{v}_i'' = \vec{v}_i - 1$ while $\vec{v}''_j = \vec{v}_j$ for $j\neq i$, then $S''$ is well-formed. Starting with the vector $\vec{v}$ and subsequently increasing or decreasing its components by $1$ we can obtain any vector $\vec{v}''$ satisfying the conditions of the lemma. Since at each step the condition of one of the claims hold, we obtain that the condition of the lemma holds in the result.
    
Recall Definitions~\ref{def:DerTree} and \ref{def:DerMultiset} of trees and multiset derivations. Denote by $T_i = (A \yields uAv)$ the $i$-th block-tree. Fix a derivation tree $T$ such that $S \ders \{T\}$ with a derivation of the tree $T$ as well.
    
    Proof of the first claim. Due to the form of $T_i \in S$, the tree $T$ contains the non-terminal~$A$. So we can glue $T_i$ into the place of some occurrence of the non-terminal $A$ and obtain the tree $T'$ as the result. $S''$ is well-formed, since $T'$ is a derivation tree of $G$ by the construction.
    
    Proof of the second claim. If $T_i$ is a subtree of $T$, than it can be removed from $T$ (as in Fig.~\ref{fig::tree:decomposition}) and the resulting tree $T'$ would also be a derivation tree of $G$. Otherwise, let us consider the steps of the fixed derivation of $T$. 
    We change this derivation as follows. Recalling that $\vec{v}_i > 1$, fix two copies $T^1_i$ and $T^2_i$ of the tree $T_i$ in the multiset $S$ such that $T^1_i$ was used earlier than $T^2_i$ in the derivation of $T$.
When $T^2_i$ must be composed with the ground tree, we skip this step. If at some step a tree $T_j = (B \yields u'B v')$ is glued into a node~$B$ of $T^2_i$, then we glue it into the corresponding node $B$ of $T^1_i$ (that is already in the ground tree). So, this modification yields a derivation $S \ders \{T', T^2_i\}$, where $T'$ is some ground tree. But then we could do the same starting from $S''$ and obtain $S'' \ders \{T'\}$. Therefore $S''$ is well-formed.    
\end{proof}

\begin{corollary}\label{cor:SprimeInPsiL}
    For any $w \in L(G) : \Lin(S'_w) \subseteq \Psi(L(G))$.
\end{corollary}

\begin{lemma}\label{lemma:wInSprime} For any $w \in L(G)$:
    $$ \Psi(w) \in  \Lin(S'_w).$$
\end{lemma}
\begin{proof}
    From the definitions it follows that $\Psi(w) = \Psi(S_w) \in \Lin(S_w)$. We prove that $\Lin(S_w) \subseteq \Lin(S'_w)$. 
Note that $\Psi(S_w) = \Psi(S'_w) + \sum_{i=1}^m (\vec{v}_i-1)\Psi(T_i)$, where $\vec{v}$ is the corresponding vector of $S_w$ and $T_i$ is the $i$-th minimal tree (in the enumeration above). By the definition a vector $\vec{u} \in \Lin(S_w)$ has the form
$$
 \vec{u} = \Psi(S_w) + \sum\limits_{T'_j \in S_w\setminus\{T_g\}} t'_j\cdot\Psi(T'_j) = \Psi(S_w) + \sum\limits_{i =1}^n t_i\cdot\Psi(T_i),
$$
where in the second equality we put together all $T'_j$'s that are copies of the same tree $T_i$. Note that each $T'_j$ is a minimal block tree, so $T'_j = T_i$ for some $i$. Thus,
\[
\phantom{wtfwtfwtfwtfw\!} \vec{u} = \left(\Psi(S'_w) + \sum_{i : \vec{v}_i > 0} (t_i+\vec{v}_i-1)\Psi(T_i)\right) \in \Lin(S'_w) \qedhere
\]
\end{proof}

\subsection{Proof of Parikh's theorem}
Let $G$ be a context-free grammar generating $L$. For a word $w \in L$ denote by $T(w)$ the set of all derivation trees $T_w$. Note that $T(w) \neq \es$ and it can be even infinite if there are $\eps$-rules in $G$. Denote by $S(w) = \{ S_w \mid \exists T_w \in T(w) : S_w \ders \{T_w\} \}$ where $S_w$ is a multiset consisting only of minimal trees (as before). Finally $S'(w) = \{ S'_w \mid S_w \in S(w)\}$; recall that $S'_w$ is the set obtained from $S_w$ by deleting the duplicates. Now we show that $S'(w)$ is finite for each $w$ and moreover the union $\cup_{w\in L} S'(w) = S'(L)$ is finite as well. Recall that there are finitely many minimal trees, and we denote them by $T_1, \ldots, T_m$. Therefore, each $S_w$ has a corresponding $m$-dimensional vector $\vec{v}$ such that
 $$\Psi(w) = \Psi(S_w) = \sum\limits_{i=1}^m \vec{v}_i\cdot\Psi(T_i).$$  
The corresponding vector $\vec{v}'$ of $S'_w$ is a 0-1 vector such that $\vec{v}'_i = 1$ iff $\vec{v}_i > 0$. So since there are only finitely many 0-1 vectors of length $m$ and each $S'_w \in S'(L)$ has a corresponding 0-1 vector $\vec{v}'$, then the set $S'(L)$ is finite as well.

Putting everything together, by Lemma~\ref{lemma:wInSprime}
$$
\Psi(L) \subseteq \bigcup_{S' \in S'(L)} \Lin(S')
$$
and by Corollary~\ref{cor:SprimeInPsiL}
$$
\bigcup_{S' \in S'(L)} \Lin(S') \subseteq \Psi(L).
$$
Since the set $S'(L)$ is finite and each set $\Lin(S')$ is linear by Eq.~\eqref{eq:LinDef}, the equality 
$$
\Psi(L) = \bigcup_{S' \in S'(L)} \Lin(S')
$$
proves Parikh's theorem.
\qed

\bibliography{parikh}

\end{document}